\documentclass{article}

\usepackage{physics}
\usepackage{amsmath}
\usepackage{tikz}
\usepackage{mathdots}
\usepackage{yhmath}
\usepackage{cancel}
\usepackage{color}
\usepackage{siunitx}
\usepackage{array}
\usepackage{multirow}
\usepackage{amssymb}
\usepackage{gensymb}
\usepackage{tabularx}
\usepackage{extarrows}
\usepackage{booktabs}
\usepackage[hyphens]{url}
\usepackage{hyperref}
\usepackage[hyphenbreaks]{breakurl}

\usetikzlibrary{fadings}
\usetikzlibrary{patterns}
\usetikzlibrary{shadows.blur}
\usetikzlibrary{shapes}

\usepackage{mathtools}
\usepackage{amsthm}
\title{The Adversary Bound Revisited: From Optimal Query Algorithms to Optimal Control}

\author{Duyal Yolcu\thanks{\href{https://github.com/qudent}{https://github.com/qudent}}}

\begin{document}
\maketitle
\newtheorem{definition}{Definition}
\newtheorem{proposition}{Proposition}
\newtheorem{thm}{Theorem}
\begin{abstract}
This note complements the paper "One-Way Ticket to Las Vegas and the Quantum Adversary" \cite{LasVegas}. I develop the ideas behind the adversary bound - universal algorithm duality therein in a different form, using the same perspective as \cite{BarnumQuantum, Barnum2007Semidefinite} in which query algorithms are defined as sequences of feasible reduced density matrices rather than sequences of unitaries. This form may be faster to understand for a general quantum information audience: It avoids defining the "unidirectional relative $\gamma_{2}$-bound" and relating it to query algorithms explicitly. It is also more general because the lower bound (and universal query algorithm) apply to a class of time-optimal control problems rather than just query problems. That is in addition to advantages to be discussed in \cite{LasVegas}, namely the more elementary algorithm and correctness proof that avoids phase estimation and spectral analysis, allows for limited treatment of noise, and improves the runtime by another $\Theta(\log(1/\varepsilon))$ compared to \cite{Lee2011Quantum,belovs2015variations}.

The approach - not new - starts with considering an optimal \textit{query} problem for state conversion - in which we are given an unknown oracle $\displaystyle L_{a}$ for $\displaystyle a\in A$, and want to construct an algorithm that transforms initial states $\displaystyle \ket{\xi _{a}}$ to acceptable target states $\displaystyle \ket{\tau _{a}}$ by invoking that oracle - as an optimal \textit{control }problem, in which $\displaystyle a\in A$ is stored as a quantum "input state" in another Hilbert space $\displaystyle \mathcal{A}$, and we want to transform $\displaystyle \ket{\xi } =\sum _{a\in A}\ket{a} \otimes \ket{\xi _{a}}$ to $\displaystyle \ket{\tau } =\sum _{a\in A}\ket{a} \otimes \ket{\tau _{a}}$ by invoking $\displaystyle L=\sum _{a\in A}\ket{a}\bra{a} \otimes L_{a}$ without being allowed to access $\displaystyle \mathcal{A}$ directly. Then we track feasible \textit{reduced density matrices }on $\displaystyle \mathcal{A}$ (which correspond to the transpose of the Gram matrices in the adversary bound literature). This information is sufficient to track the system's state by standard facts on purifications and their unitary equivalence. This is close to \cite{BarnumQuantum, Barnum2007Semidefinite}.

The adversary bound is the inverse maximal speed we can achieve in reduced density matrix space from any starting RDM, travelling in the desired direction. We describe a new universal control algorithm that matches this speed up to an error-dependent factor by slightly perturbing initial and target state by that ideal starting RDM; by a linearity argument, an algorithm going along a straight line between these perturbed states is feasible. We can then bound the error in the final state if we apply that algorithm to the original initial state, rather than the perturbed one.

Importantly, this approach doesn't assume that $\displaystyle L$ is "read-only", i.e. block-diagonal in $\displaystyle \mathcal{A}$, anymore - as long as there is still an "idle subspace" as defined in the note. We can therefore apply it to problems in which the register $\displaystyle \mathcal{A}$ is to be manipulated, rather than just read out. The argument also works with problems in which $\displaystyle L$ is only subunitary, i.e. may correspond to noise occuring and the algorithm "giving up" in some instances.

To follow this text, one only needs basic knowledge in quantum physics including reduced density operators, purifications and their local equivalence. Therefore, I hope it has expository value for a general quantum information audience that wants to understand adversary bound - universal query algorithm dualities.
\end{abstract}
\section*{Acknowledgments}
I thank Alexander Belov (Aleksandrs Belovs) and Berare Göktürk for helpful discussions while writing the draft.
\tableofcontents
\section{The problem - discrete time}
\label{sec:problem}
We consider a tripartite quantum system $\mathcal{ABC} ,$ assumed to be finite-dimensional for convenience. We want to construct an algorithm that transforms an initial state $\displaystyle \ket{\xi }\neq 0 \in \mathcal{ABC}$ to a target state $\displaystyle \ket{\tau } \in \mathcal{ABC}$ (or similar, e.g. allowing a range of acceptable target states) as fast as possible. The catch is that the algorithm is not allowed to arbitrarily act on $\displaystyle \mathcal{A}$ - instead, there is a fixed subunitary interaction operator $\displaystyle L$ (i.e. fulfilling $\displaystyle \left\Vert L\ket{\varphi }\right\Vert \leq \left\Vert \ket{\varphi }\right\Vert$ for any $\ket{\varphi}$) that acts on $\displaystyle \mathcal{AB}$ (not $\mathcal{C}$) once per timestep. On the other hand, we are allowed to apply arbitrarily complex unitaries on $\displaystyle \mathcal{BC}$ at any time, without cost. We also assume that $\displaystyle \mathcal{C}$ is "as large as the algorithm could need it to be"; in the proof of Proposition \ref{prop:densitymatrixsequences_equal_queryalgorithms}, we'll see that  $\displaystyle \dim\mathcal{C} \geq \dim(\mathcal{AB})$ works for any algorithm. We optimize over the number of timesteps (or, conversely, lower-bound the number of timesteps necessary).

Finally, we fix a designated normalized state $\displaystyle \ket{\mathrm{idle}} \in \mathcal{B}$. We assume that $\displaystyle L$ acts trivially on $\displaystyle \ket{\mathrm{idle}}$, i.e. $\displaystyle L\left(\ket{\varphi }\ket{\mathrm{idle}}\right) =\ket{\varphi }\ket{\mathrm{idle}}$ for all $\displaystyle \ket{\varphi } \in \mathcal{A} .$ This is equivalent to assuming that $\displaystyle L$ can be applied in a "controlled" way, and the algorithm can always choose to do nothing.

Allowing subunitary, rather than only unitary, $\displaystyle L$ allows a limited discussion of \textbf{noise:} If the true transformation a physical system undergoes involves a noisy quantum channel, one can choose $\displaystyle L$ to be one Kraus operator of that quantum channel and consider it the "successful" Kraus operator. The other Kraus operators can be considered "errors", and one stops tracking the computation in case one of them is applied. Then the norm of the state will decay over time - corresponding to losing probability mass in case an error occurs - but one can still define sets of acceptable, sub-normalized target states, and achieving them with that $\displaystyle L$ is sufficient to solve the entire problem. However, it is necessary for our universal algorithm that $\displaystyle L$ preserves the norm when acting on $\displaystyle \ket{\mathrm{idle}} .$

As mentioned in the abstract, we can encode a query complexity problem by choosing $\displaystyle L$ block-diagonal. For a more thorough introduction to query problems, see \cite{belovs2015variations}.
\section{Quantum algorithms as sequences of feasible reduced density matrices on AB}
\label{sec:quantum-feasible-rdms}
We now define quantum algorithms in terms of the feasible intermediate reduced density matrices on $\displaystyle \mathcal{AB}$. This section uses essentially the same ideas as \cite{BarnumQuantum, Barnum2007Semidefinite} (aside from pointing out that they apply to a wider class of $\displaystyle L$).
\begin{definition}\label{def:QuantumAlgRDMSequence}
In our discussion, a T-step \textbf{quantum algorithm }is a list of positive semidefinite operators $\displaystyle \left( \pi ^{0} ,\pi ^{1} ,\dotsc ,\pi ^{T-1}\right)$ representing (non-normalized) reduced density matrices on $\displaystyle \mathcal{AB}$ such that for all $\displaystyle j$,
\begin{equation}\label{eq:consistency_density_matrix_list}
\mathrm{tr}_{\mathcal{B}} \pi ^{j+1} =\mathrm{tr}_{\mathcal{B}}\left( L \pi ^{j} L^{\dagger }\right) .
\end{equation}
\end{definition}
\begin{proposition}\label{prop:densitymatrixsequences_equal_queryalgorithms}
In the model of Section \ref{sec:problem}, it is possible to transform $\displaystyle \ket{\xi } \mapsto \ket{\tau }$ in $\displaystyle T$ timesteps iff such a list exists with $\displaystyle \mathrm{tr}_{\mathcal{BC}}\ket{\xi }\bra{\xi } =\mathrm{tr}_{\mathcal{B}} \pi ^{0}$ and $\displaystyle \mathrm{tr}_{\mathcal{BC}}\ket{\tau }\bra{\tau } =\mathrm{tr}_{\mathcal{B}}\left( L \pi ^{T-1} L^{\dagger }\right) .$
\end{proposition}
\begin{proof}
First suppose that such a transformation is possible and let $\displaystyle \ket{\Phi ^{j}}\in\mathcal{ABC}$ be the system's state directly before the $\displaystyle j+1$th application of $\displaystyle L$ (counting from $\displaystyle 1$). Then set $\displaystyle \pi ^{j} := \mathrm{tr}_{\mathcal{C}}\ket{\Phi ^{j}}\bra{\Phi ^{j}}.$ As reduced density operators of a nonnormalized state, these are positive semidefinite. Directly after the $\displaystyle j+1$th application of $\displaystyle L,$ the reduced density matrix on $\displaystyle \mathcal{A}$ is $\displaystyle \mathrm{tr}_{\mathcal{B}}\left( L \pi ^{j} L^{\dagger }\right)$ by standard quantum physics. Similarly, directly before the $\displaystyle j+2$nd application, the RDM is $\displaystyle \mathrm{tr}_{\mathcal{B}}\left( \pi ^{j+1}\right) .$ But these matrices must be equal because between the $\displaystyle j+1$th and the $\displaystyle j+2$nd application, the quantum computer may only act on $\displaystyle \mathcal{BC} ,$ and not on $\displaystyle \mathcal{A} .$

Conversely, suppose there is an algorithm as in Definition \ref{def:QuantumAlgRDMSequence} and consider a sequence of purifications of the $\displaystyle \pi ^{j}$ on $\displaystyle \mathcal{C} ,$ i.e. $\displaystyle \ket{\Phi ^{j}} \in \mathcal{ABC}$ such that $\displaystyle \mathrm{tr}_{\mathcal{C}}\ket{\Phi ^{j}}\bra{\Phi ^{j}} =\pi ^{j} .$ By standard quantum physics again, these always exist if $\displaystyle \dim\mathcal{AB} \leq \dim\mathcal{C}$ \cite{nielsen2002quantum}. Then by Equations \ref{eq:consistency_density_matrix_list}, $\displaystyle ( L\otimes I_{\mathcal{C}})\ket{\Phi ^{j}}$ and $\displaystyle \ket{\Phi ^{j+1}}$ are purifications of the same reduced density matrix on $\displaystyle \mathcal{A}.$ The same is true for $\displaystyle \ket{\Phi ^{0}}$ and $\displaystyle \ket{\xi } ,$ as well as $\displaystyle ( L\otimes I_{\mathcal{C}})\ket{\Phi ^{T-1}}$ and $\displaystyle \ket{\tau } .$ As all such purifications are related by local unitaries (i.e. unitaries acting only on $\displaystyle \mathcal{BC}$) \cite{nielsen2002quantum, OzolsUnitaryEquivalence}, a valid quantum algorithm exists that starts with $\displaystyle \ket{\xi }$ and applies these connecting unitaries between applications of $\displaystyle L.$
\end{proof}
As promised in the abstract, this argument essentially works by tracking the reduced density matrix on $\displaystyle \mathcal{A} .$ This set of lists of operators is also a convex set in a natural way, which gives rise to nice properties - see \cite{LasVegas} for details.
\section{Adversary bound}
Now consider a T-step quantum query algorithm transforming $\displaystyle \ket{\xi } \mapsto \ket{\tau }$ and consider the sum of all $\displaystyle \pi ^{j}$,
\begin{equation}
\overline{\pi } :=\sum _{j=0}^{T-1} \pi ^{j} \in \mathbb{S}\mathcal{_{\mathcal{AB}}} ,
\end{equation}
where $\displaystyle \mathbb{S}_{\mathcal{AB}}$ denotes the set of positive semidefinite operators on $\displaystyle \mathcal{AB} .$ Then\begin{gather}
\mathrm{tr_{\mathcal{B}}\left( L\overline{\pi } L^{\dagger }\right) =}\sum _{j=0}^{T-1}\mathrm{tr}_{\mathcal{B}}\left( L\pi ^{j} L^{\dagger }\right) =\sum _{j=0}^{T-2}\mathrm{tr}_{\mathcal{B}} \pi ^{j+1} +\mathrm{tr}_{\mathcal{BC}}\left(\ket{\tau }\bra{\tau }\right)\notag\\
=\mathrm{tr_{\mathcal{B}}\overline{\pi } +tr}_{\mathcal{BC}}\left(\ket{\tau }\bra{\tau } -\ket{\xi }\bra{\xi }\right) . \label{eq:evol_pibar}
\end{gather}
Furthermore,
\begin{equation}
\mathrm{tr}(\overline{\pi }) \leq \ T\ \bra{\xi }\ket{\xi}:\label{eq:boundtrpibar}
\end{equation}
If we turn the sequence of $\displaystyle \pi ^{j}$ into a quantum algorithm involving a sequence of $\displaystyle \ket{\Phi ^{j}}$ as in \ref{sec:quantum-feasible-rdms}, $\displaystyle \mathrm{tr}\left( \pi ^{j}\right) =\left\Vert \ket{\Phi ^{j}}\right\Vert ^{2} \leq \left\Vert \ket{\xi }\right\Vert ^{2}$ by subunitarity of $\displaystyle L$; Inequality \ref{eq:boundtrpibar} is the result of adding these inequalities.

If a $\displaystyle T$-query algorithm exists, \textit{some }$\displaystyle \overline{\pi }$ must exist, which yields the following bound:\footnote{The motivation for this term is clearer in other expositions, such as Childs's lecture notes \cite{childs2017lecture}.}
\begin{definition}[Adversary bound]\label{def:adversary}The \textbf{\textit{adversary bound}} of a state conversion problem, denoted $\displaystyle \mathrm{Adv}\left(\ket{\xi }\rightarrow \ket{\tau }\right)$ with $\ket{\xi}\neq 0$, is the optimal value of the minimization problem (which we call the \textbf{primal problem})
\begin{align}\label{eq:primal_start}
\mathrm{minimize} & \ \mathrm{tr} \ (\overline{\pi }) /\bra{\xi }\ket{\xi }\\
\mathrm{subject\ to} & \ \mathrm{tr}_{\mathcal{B}}\left( L\overline{\pi } L^{\dagger } -\overline{\pi }\right) =\mathrm{tr}_{\mathcal{BC}}\left(\ket{\tau }\bra{\tau } -\ket{\xi }\bra{\xi }\right) \label{eq:primal_filter},\\
 & \ \overline{\pi } \in \mathbb{S}_{\mathcal{AB}} .\label{eq:primal_end}
\end{align}
By the discussion above, $\displaystyle \mathrm{Adv}\left(\ket{\xi }\rightarrow \ket{\tau }\right)$ lower-bounds the number of queries of quantum query algorithms solving the state conversion problem exactly.
\end{definition}
The inverse of this problem's optimal solution is also the answer to the question "In the space of reduced density matrices on $\displaystyle \mathcal{A} ,$ what is the maximum fraction of the desired change $\displaystyle \mathrm{tr}_{\mathcal{BC}}\left(\ket{\tau }\bra{\tau } -\ket{\xi }\bra{\xi }\right)$ achievable, starting from \textit{any} state with the correct normalization?"

Subsection \ref{sec:control-to-query} discusses a strengthening relevant for query problems, omitted here to reduce technicality.

The following remark uses semidefinite programming duality; the result isn't necessary for the remainder of the discussion, and a reader unfamiliar with the technique may take it on faith. The optimal value of the problem is lower-bounded by the optimal value of the maximization problem (the dual problem)
\begin{align}\label{eq:dual_start}
\mathrm{maximize} & \left(\bra{\tau } \Gamma \otimes I_{\mathcal{BC}}\ket{\tau } -\bra{\xi } \Gamma \otimes I_{\mathcal{BC}}\ket{\xi }\right) /\bra{\xi }\ket{\xi }\\
\mathrm{subject\ to} & \ L^{\dagger }( \Gamma \otimes I_{\mathcal{B}}) L-\Gamma \otimes I_{\mathcal{B}} \preceq I_{\mathcal{AB}} ,\\
 & \ \Gamma \in \mathbb{H}_{\mathcal{A}} ,\label{eq:dual_end}
\end{align}
where $\displaystyle \mathbb{H}_{\mathcal{A}}$ denotes the space of Hermitian matrices on $\displaystyle \mathcal{A}$. This means that finding any feasible $\displaystyle \Gamma $ for this problem corresponds to a proof that no algorithm can be faster - which is more convenient for finding lower bounds on the number of steps necessary for a conversion. We can see that \textit{Slater's} \textit{strong duality condition }is fulfilled by choosing $\displaystyle \Gamma =0$ in the dual problem. This means that the best solution to Problem \ref{eq:dual_start}-\ref{eq:dual_end} results in a value equal to the best solution of Problem \ref{eq:primal_start}-\ref{eq:primal_end}.
\section{Matching the lower bound by a universal algorithm}
Now assume we have some feasible solution $\displaystyle \overline{\pi }$ of the optimization problem in Definition \ref{def:adversary}, which doesn't have to be optimal. Can we "turn it around" and obtain an algorithm to transform $\displaystyle \ket{\xi }\rightarrow \ket{\tau }$ in $\displaystyle \mathrm{tr} \ \overline{\pi }$ steps? The answer will turn out to be "almost".
\begin{proposition}
Using the notations above, for any integer $\displaystyle T' >0$, the sequence of $\displaystyle \pi ^{j}$
\begin{align*}
\left( \pi ^{j}\right)_{0\leq j< T'} & =\left(\left(\frac{T'-j}{T'} \ \mathrm{tr}_{\mathcal{BC}}\ket{\xi }\bra{\xi } +\frac{j}{T'} \ \mathrm{tr}_{\mathcal{BC}}\ket{\tau }\bra{\tau }\right) \otimes \ket{\mathrm{idle}}\bra{\mathrm{idle}} +\frac{\overline{\pi }}{T'}\right)_{0\leq j< T'}
\end{align*}
constitutes a $\displaystyle T'$-query quantum query algorithm solving the state conversion problem $\displaystyle \ket{\xi } \otimes \ket{0} +\frac{\ket{v}}{\sqrt{T'}} \otimes \ket{1} \mapsto \ket{\tau } \otimes \ket{0} +\frac{\ket{v}}{\sqrt{T'}} \otimes \ket{1} ,$ where $\displaystyle \ket{v}$ is a purification of $\displaystyle \overline{\pi }$ on $\displaystyle \mathcal{C}$ (and we add a qubit to the ancilla space).
\end{proposition}
Intuitively, this algorithm works by using a scaled down version of $\overline{\pi}$ to go from initial to target density matrix, along a shifted straight line in the space of reduced density matrices --- by Equation \ref{eq:primal_filter}, $\overline{\pi}$ is able to induce exactly the change in the desired direction.
\begin{proof}
We first show that $\left( \pi ^{j}\right)_{0\leq j< T'} $ is a quantum query algorithm in the sense of Definition \ref{def:QuantumAlgRDMSequence}. This is equivalent to the conditions that
\begin{itemize}
\item Each $\pi^j\succeq 0$: This follows from the facts that density matrices are positive semidefinite and closed under convex mixtures.
    \item Equation \ref{eq:consistency_density_matrix_list} holds for $j=0,\ldots,T-1$. Plugging in our $\pi^j$, we need to show that
\begin{multline*}
        \mathrm{tr}_{\mathcal{B}}\left(\left(\frac{T'-j-1}{T'} \ \mathrm{tr}_{\mathcal{BC}}\ket{\xi }\bra{\xi } +\frac{j+1}{T'} \ \mathrm{tr}_{\mathcal{BC}}\ket{\tau }\bra{\tau }\right) \otimes \ket{\mathrm{idle}}\bra{\mathrm{idle}} +\frac{\overline{\pi }}{T'}\right)\\
        = \mathrm{tr}_{\mathcal{B}}\left( L \left(\left(\frac{T'-j}{T'} \ \mathrm{tr}_{\mathcal{BC}}\ket{\xi }\bra{\xi } +\frac{j}{T'} \ \mathrm{tr}_{\mathcal{BC}}\ket{\tau }\bra{\tau }\right) \otimes \ket{\mathrm{idle}}\bra{\mathrm{idle}} +\frac{\overline{\pi }}{T'}\right) L^{\dagger }\right) 
    \end{multline*}
    for these $j$. We know that $L$ acts trivially on $\ket{\mathrm{idle}}$, so the condition is equivalent to
    \begin{multline*}
        \left(\frac{T'-j-1}{T'} \ \mathrm{tr}_{\mathcal{BC}}\ket{\xi }\bra{\xi } +\frac{j+1}{T'} \ \mathrm{tr}_{\mathcal{BC}}\ket{\tau }\bra{\tau }\right) + \mathrm{tr}_{\mathcal{B}}\left(\frac{\overline{\pi }}{T'}\right)\\
        = \left(\frac{T'-j}{T'} \ \mathrm{tr}_{\mathcal{BC}}\ket{\xi }\bra{\xi } +\frac{j}{T'} \ \mathrm{tr}_{\mathcal{BC}}\ket{\tau }\bra{\tau }\right) + \mathrm{tr}_{\mathcal{B}}\left( L \frac{\overline{\pi }}{T'} L^{\dagger }\right).
    \end{multline*}
    Rearranging the terms, we transform this condition into
    \begin{equation*}
        \frac{1}{T'} \ \mathrm{tr}_{\mathcal{BC}} \left(\ket{\tau }\bra{\tau } - \ket{\xi }\bra{\xi } \right) 
        = \frac{1}{T'}\mathrm{tr}_{\mathcal{B}}\left( L \overline{\pi }L^{\dagger } - \overline{\pi}\right).
    \end{equation*}
    Up to a factor, this is exactly the condition of Equation \ref{eq:primal_filter}, and we are guaranteed it is fulfilled for a feasible $\overline{\pi}$.
\end{itemize}
Similarly, we may show that
\begin{align}
    \mathrm{tr}_{\mathcal{B}} \left( \pi^0 \right) &= \mathrm{tr}_{\mathcal{BC}}\ket{\xi }\bra{\xi } + \frac{\mathrm{tr}_{\mathcal{B}}\left(\overline{\pi }\right)}{T'},\\
    \mathrm{tr}_{\mathcal{B}} \left( L \pi^{T'-1} L^\dagger\right) &= \mathrm{tr}_{\mathcal{BC}}\ket{\tau }\bra{\tau } + \frac{\mathrm{tr}_{\mathcal{B}}\left(\overline{\pi }\right)}{T'},
\end{align}
which equals the reduced density matrices of the claimed initial and final states. By Proposition \ref{prop:densitymatrixsequences_equal_queryalgorithms}, this implies that the claimed transformation is indeed possible.
\end{proof}
As $\displaystyle T'\rightarrow \infty $, initial and target states converge to our desired $\displaystyle \ket{\xi } ,\ket{\tau }$ (apart from having redefined the ancilla space. Let $E\colon\mathcal{ABC}\to\mathcal{ABC}$ be the total effective evolution operator applied by the algorithm; this operator must be subunitary. If we apply the algorithm to our true initial state $\displaystyle \ket{\xi } \otimes \ket{0} ,$ and project to the $\displaystyle \ket{0}$ subspace in the end (i.e. measures that qubit and outputs "failure" in case the result is $\displaystyle 1$, generally reducing the norm of the state), the resulting state fulfills
\begin{gather}
\left\Vert P_{0} E\left(\ket{\xi } \otimes \ket{0}\right) -\ket{\tau } \otimes \ket{0}\right\Vert /\left\Vert \ket{\xi }\right\Vert  \notag\\
=\left\Vert P_{0} E\left(\ket{\xi } \otimes \ket{0} +T^{\prime -1/2}\ket{\nu } \otimes \ket{1}\right) -P_{0} E\left( T^{\prime -1/2}\ket{\nu } \otimes \ket{1}\right) -\ket{\tau } \otimes \ket{0}\right\Vert /\left\Vert \ket{\xi }\right\Vert  \notag\\
=\left\Vert \ket{\tau } \otimes \ket{0} -P_{0} E\left( T^{\prime -1/2}\ket{\nu } \otimes \ket{1}\right) -\ket{\tau } \otimes \ket{0}\right\Vert /\left\Vert \ket{\xi }\right\Vert  \notag\\
=T^{\prime -1/2}\left\Vert -P_{0} E\left(\ket{\nu } \otimes \ket{1}\right)\right\Vert \leq T^{\prime -1/2}\sqrt{\mathrm{tr}\overline{\pi }} /\left\Vert \ket{\xi }\right\Vert \label{eq:bounderror},
\end{gather}
where we used subunitarity of $E$ and $\displaystyle P_{0}$ in the last line. In fact, a similar argument would show that the norm difference would be at most twice that if we were not allowed to throw away part of the state in the last step. For an optimal $\displaystyle \overline{\pi } ,$ $\displaystyle \sqrt{\mathrm{tr}\overline{\pi }} /\left\Vert \ket{\xi }\right\Vert =\sqrt{\mathrm{Adv}\left(\ket{\xi } \mapsto \ket{\tau }\right)} .$

For large $\displaystyle T',$ each individual step puts most of its amplitude into the $\displaystyle \ket{\mathrm{idle}}$ subspace. Remarkably, one can show (proof omitted) that for the algorithm's intermediate states $\displaystyle \ket{\Phi ^{j}} ,$ the quantity $\displaystyle \sum _{j=0}^{T'-1}\bra{\Phi ^{j}} I-P_{\mathrm{idle}}\ket{\Phi ^{j}} /\bra{\xi }\ket{\xi } \leq \mathrm{tr\overline{\pi } /\bra{\xi }\ket{\xi }}$ independent of $\displaystyle T'.$ The analogue for query problems - called "Las Vegas complexity" - is defined and studied in \cite{LasVegas}.

In conclusion:
\begin{thm}~
\begin{enumerate}
\item A control algorithm converting $\displaystyle \ket{\xi }$ to $\displaystyle \ket{\tau }$ uses at least $\displaystyle \mathrm{Adv}\left(\ket{\xi }\rightarrow \ket{\tau }\right)$ steps,
\item Conversely, for any acceptable error $\displaystyle \varepsilon $, we can find an algorithm converting 
$\displaystyle \ket{\xi } \otimes \ket{0}$ to $\displaystyle \ket{\tau '} \otimes \ket{0} +\ket{\Delta } \otimes \ket{1}$, with $\displaystyle \left\Vert \ket{\tau '} -\ket{\tau }\right\Vert /\left\Vert \ket{\xi }\right\Vert \leq \varepsilon ,$ that takes\begin{equation}
T' =\left\lceil \frac{\mathrm{Adv}\left(\ket{\xi }\rightarrow \ket{\tau} \right)}{\varepsilon ^{2}}\right\rceil 
\end{equation}
steps and fulfills $\displaystyle \sum _{j=0}^{T-1}\bra{\Phi ^{j}} I-P_{\mathrm{idle}}\ket{\Phi ^{j}} /\bra{\xi }\ket{\xi } \leq \mathrm{Adv}\left(\ket{\xi }\rightarrow \ket{\tau }\right)$ on the intermediate states.
\end{enumerate}
\end{thm}
As remarked in the abstract, this algorithm corresponds to going along a straight line with constant velocity in the space of reduced density operators. As $\displaystyle \overline{\pi }$ is not "used up" during this transformation, we can interpret it as a "catalyst" in the spirit of catalytic states in LOCC transformations (see \cite{quantum_catalyst}).
\section{Further remarks}
\subsection{From control to query algorithms}\label{sec:control-to-query}
I briefly discuss how to modify this argument for quantum query complexity problems in state conversion problems; I skipped this before for simplicity. This note completely ignores function evaluation and output conditions - i.e. the question of what final states allow calculating some function of the input in a query problem. See e.g. \cite{LasVegas}, \cite{belovs2015variations}, \cite{Barnum2007Semidefinite} for a more thorough discussion of query complexity problems.

Start directly after Equation \ref{eq:evol_pibar}. Let $\displaystyle P_{\mathcal{A} '}$ be a projector onto a subspace $\displaystyle \mathcal{A} '\subseteq \mathcal{A}$ such that $\displaystyle P_{\mathcal{A} '} L=LP_{\mathcal{A} '} .$ Choosing $\displaystyle \mathcal{A'=A}$ and $\displaystyle P_{\mathcal{A} '} =I$ will always work; when dealing with a query problem and $\displaystyle L$ is block-diagonal in some basis $\displaystyle \left\{\ket{a}\right\}_{a\in A}$ of $\displaystyle \mathcal{A}$, we could choose $\displaystyle \mathcal{A} ':=\mathrm{span} \left\{ \ket{a} \right\}$ as well for any $\displaystyle a\in A.$ The argument that shows $\displaystyle \mathrm{tr}(\overline{\pi }) \leq \ T\bra{\xi }\ket{\xi }$ (Inequality \ref{eq:boundtrpibar}) is in fact sufficient to show that
\begin{equation}
\mathrm{tr}( P_{\mathcal{A} '}\overline{\pi }) \leq \ T\ \bra{\xi } P_{\mathcal{A} '}\ket{\xi }
\end{equation}

for any such $\displaystyle \mathcal{A} ',$ because we can commute $\displaystyle P_{\mathcal{A} '}$ through the entire evolution.

So each suitable $\displaystyle \mathcal{A} '$ yields a lower bound on $\displaystyle T,$ and we can replace the optimization target $\displaystyle \mathrm{tr}\overline{\pi }/\bra{\xi }\ket{\xi }$ given in Definition \ref{def:adversary} by
\begin{equation}
\underset{\mathcal{A} '\subseteq \mathcal{A}\colon P_{\mathcal{A} '} L=LP_{\mathcal{A} '}, P_{\mathcal{A} '}\ket{\xi}\neq 0}{\mathrm{sup}} \ \left(\mathrm{tr}( P_{\mathcal{A} '}\overline{\pi }) /\bra{\xi } P_{\mathcal{A} '}\ket{\xi }\right) 
\end{equation}
and add the constraint
\begin{equation}
\mathrm{tr}\left( P_{\mathcal{A'}} \overline{\pi}\right) = 0
\end{equation}
for any $\mathcal{A}'$ such that $P_{\mathcal{A}'}\ket{\xi}=0.$

We can also fix a set of $\displaystyle \mathcal{D'}$ that fit, and consider the optimization problem that considers only these. For a block-diagonal $\displaystyle L$ as above and $\left\Vert P_a \ket{\xi} \right\Vert=1$ for all $a,$ this results in an optimization problem equivalent to the unidirectional relative $\displaystyle \gamma _{2}$-bound of \cite{LasVegas}.

Conversely, suppose we have a optimal solution of that modified optimization problem. Then we can insert any $\displaystyle P_{\mathcal{A} '}$ with $P_{\mathcal{A} '}\ket{\xi}\neq 0$ into the derivation of Inequality \ref{eq:bounderror}. Using the fact that it commutes with all operators involved in that derivation, we derive that
\begin{equation}
\frac{\left\Vert P_{\mathcal{A} '} P_{0} A\left(\ket{\xi } \otimes \ket{0}\right) -P_{\mathcal{A} '}\ket{\tau } \otimes \ket{0}\right\Vert}{\left\Vert P_{\mathcal{A} '}\ket{\xi }\right\Vert} \leq \sqrt{\frac{\mathrm{Adv}\left(\ket{\xi } \mapsto \ket{\tau }\right)}{T'}}
\end{equation}
for each individual $\displaystyle P_{\mathcal{A} '} ,$ rather than just $\displaystyle P_{\mathcal{A} '} =I.$ In , this allows us to consider the error bound $\displaystyle \left\Vert \ket{\tau '} -\ket{\tau }\right\Vert /\left\Vert \ket{\xi }\right\Vert \leq \varepsilon $ with $\displaystyle \left\Vert P_{\mathcal{A} '}\left(\ket{\tau '} -\ket{\tau }\right)\right\Vert /\left\Vert P_{\mathcal{A} '}\ket{\xi }\right\Vert \leq \varepsilon $ for each individual ones. If we have considered a query problem as a control problem as in the abstract, and want to ensure that the error in the state conversion is small for all possible inputs, such a strengthening is necessary.
\subsection{Continuous time}
This note discusses everything in discrete time; however, quantum physics as we know it is continuous and described by differential equations. In a physical system, the interaction between $\displaystyle \mathcal{A}$ and $\displaystyle \mathcal{B}$ would be described by a Hamiltonian $\displaystyle H;$ we may model a situation in which the wavefunction may decohere, and we stop considering the decohered parts, by choosing a non-Hermitian $\displaystyle H.$

One approach to bridging the gap is to choose $\displaystyle \epsilon  >0$ and consider a discrete-time query model with $\displaystyle L_{\epsilon } :=e^{-iH \epsilon } .$ Then $\displaystyle T$ steps correspond to an elapsed time $\displaystyle T\epsilon .$ Intuitively, the associated family of lower bounds and algorithms should converge to a description of the continuous-time situation as $\displaystyle \epsilon \mapsto 0^{+} .$ However, I didn't succeed in making all associated analysis statements rigorous.
\subsection{Other characterizations of quantum processes}
As mentioned, Section \ref{sec:quantum-feasible-rdms} is very similar to the semidefinite programming (SDP) characterization of quantum algorithms by \cite{BarnumQuantum, Barnum2007Semidefinite}. Incidentally, an SDP characterization of the success probability is also possible if the transformations aren't subunitaries, but arbitrary quantum channels between mixed states (e.g. because they introduce errors). This can be done by an application of the frameworks developed independently in \cite{Gutoski2007Toward, Chiribella2009Theoretical}. However, the matrix size necessary here is exponential in $\displaystyle T.$

In continuous time, \cite{Khaneja2001Time} discuss time-optimal control in a still more general setting based on the Pontryagin maximum principle.
\section{Conclusion and outlook}
The main novelty in this note is the universal algorithm, which is simpler and more general than the previous one based on phase detection \cite{Lee2011Quantum, belovs2015variations} and shaves another factor of $\Theta\left(\log(1/\varepsilon)\right)$ off the runtime. The way we obtained this algorithm, and proved its correctness, is also unusual:
\begin{itemize}
\item Instead of specifying gates and families of states directly, we considered all inputs at once in an associated control problem and feasible ways to manipulate reduced density matrices involving these inputs,
\item Instead of proving correctness starting with the correct initial state, and showing that the final state is not too wrong after application of the algorithm, we started with a slightly wrong initial state, and proved that the final state will be correct when applied to that modified state.
\end{itemize}

These ideas may be useful to devise other quantum algorithms.

A query-efficient algorithm doesn't necessarily translate into a gate-efficient one in the usual model of quantum complexity, as the algorithm's unitaries may be hard to construct. For example, the query complexity of the $\displaystyle k$-distinctness problem was characterized by Belovs in 2012 \cite{belovskdistinctness} using the adversary method, but an algorithm matching this complexity (up to a polylogarithmic factor) was only presented in 2022 by Jeffery and Zur \cite{Jeffery2022Multidimensional}. So it would be interesting to find conditions that $\displaystyle \overline{\pi }$ needs to fulfill so that the unitaries involved in the associated universal algorithm are efficiently representable.

Though it is not obvious from the presentation, perhaps the closest relative to the algorithm presented here is the adiabatic algorithm given by Brandeho and Roland \cite{brandeho2015universal} in the continuous-time setting. In particular, they use the idea of slowly moving from modified initial to modified target states as well, and their algorithm has the same $\Theta(\log(1/\varepsilon))$ speedup compared to \cite{Lee2011Quantum}. The essential difference to the algorithm presented here is that we avoid any error term for the intermediate steps of the computation --- while, for finite runtime, an adiabatic algorithm incurs a nonzero error during the computation as well. So, conversely, it may be worth investigating which Gram matrix evolutions more concrete adiabatic algorithms correspond to, and attempting to optimize them based on the results.
\section{Further references}
The adversary method for quantum query algorithms has evolved over multiple decades from the BBBV lower bound on Grover's search problem \cite{Bennett1997Strengths}; after being defined in \cite{ambainis2000quantum}, \cite{spalek2004all, 
hoyer2007negative,
Lee2011Quantum, reichardt2011reflections} were some contributions along the way. The previous most general, "state of the art" discussion is \cite{belovs2015variations}; a more pedagogical one in \cite{childs2017lecture}.
\bibliographystyle{plain}
\bibliography{main} 
\end{document}